\documentclass[pra,aps,twocolumn,superscriptaddress,showpacs]{revtex4}
\usepackage{graphicx,graphics,epsfig}
\usepackage{dcolumn}
\usepackage{bm}
\usepackage{amsmath,calc}
\usepackage{verbatim}
\usepackage{color}
\usepackage[colorlinks=false]{hyperref} 
\usepackage{subfigure}
\usepackage{times,natbib}
\usepackage{amsmath,amsfonts,amssymb,graphics,graphicx,epsfig,color,times,natbib}
\usepackage{tikz}
\usepackage{verbatim}
\usepackage[standard]{ntheorem}

\newcommand{\tr}{\mbox{Tr}}

\newcommand{\arc}{\mbox{Arc}}

\newcommand\bigfrown[2][\textstyle]{\ensuremath{%
  \array[b]{c}\text{\resizebox{3ex}{.7ex}{$#1\frown$}}\\[-0.9ex]#1#2\endarray}}

\date{\today}

\begin{document}
\title{Quantum contextuality of a qutrit state}

\author{Zhen-Peng Xu}
 \affiliation{Theoretical Physics Division, Chern Institute of Mathematics, Nankai University,
 Tianjin 300071, People's Republic of China}

\author{Hong-Yi Su}
\email{hysu@mail.nankai.edu.cn}
 \affiliation{Theoretical Physics Division, Chern Institute of Mathematics, Nankai University,
 Tianjin 300071, People's Republic of China}

\author{Jing-Ling Chen}
\email{chenjl@nankai.edu.cn}
 \affiliation{Theoretical Physics Division, Chern Institute of Mathematics, Nankai University,
 Tianjin 300071, People's Republic of China}
 \affiliation{Centre for Quantum Technologies, National University of Singapore,
 3 Science Drive 2, Singapore 117543}

\begin{abstract}
We present a study of quantum contextuality of three-dimensional mixed states for the Klyachko-Can-Binicio\u{g}lu-Shumovsky (KCBS) and the Kurzy\'{n}ski-Kaszlikowski (KK) noncontextuality inequalities. For any class of states whose eigenvalues are arranged in decreasing order, a universal set of measurements always exists for the KK inequality, whereas none does for the KCBS inequality. This difference can be reflected from the spectral distribution of the overall measurement matrix. Our results would facilitate the error analysis for experimental setups, and our spectral method in the paper combined with graph theory could be useful in future studies on quantum contextuality.
\end{abstract}

 \pacs{03.65.Ud,
03.67.Mn,
42.50.Xa}

\maketitle

\section{Introduction}
Quantum contextuality, which was independently discovered by Kochen and Specker (KS) \cite{KS}, and Bell \cite{bell66}, is a fundamental concept in quantum information theory. It can be revealed by the KS sets in a logic manner, or by violation of statistical noncontextuality inequalities. So far, many theoretical and experimental works have been accomplished in order to find the optimal noncontextuality inequalities \cite{kcbs,Yu-Oh,Cabello5,kk9ray} or KS sets \cite{Peres1991,ck1993,kern1994,Cabello1,kp1995}, which in turn contribute to the use of speeding up quantum algorithms \cite{speedup}.

New theoretical tools \cite{CSW14,cohom} have been invented to study contextuality. Graph theory is a such representative that finds its wide applications and effectiveness in discussing the contextual behavior~\cite{CSW14}. Two different rank-$1$ projective measurements $P_1$ and $P_2$ are said to be exclusive if they commute with one another. The exclusivity relation of a set of rank-$1$ measurements $P_i$'s for a noncontextuality inequality $\sum_i \langle P_i \rangle \le \alpha$ can be effectively represented in an exclusivity graph $G$ consisting of vertices and edges, where a pair of vertices $i,j$ are connected if and only if the corresponding events of probability $P_i, P_j$ are mutually exclusive. For each exclusivity graph, the classical bound $\alpha$ of the noncontextuality inequality equals to the independence $\alpha(G)$, and the maximal quantum prediction is just the Lov\'{a}sz number $\vartheta(G)$ \cite{Lovasz79}.

Nevertheless, there is no perfectly ``pure" state in actual experiment. It is quite necessary to consider mixed states and analyze their influence upon contextuality. Although there are state-independent noncontextuality (SIC) inequalities, whose quantum violation is independent of which state is to be measured, yet in general the violation of a noncontextuality inequality may depend on the mixedness of the state. There have been proposed various measures of the mixedness of a state, linear entropy \cite{lineare} among them is an efficient one easy to compute: For a $d$-dimensional mixed state $\rho$, the linear entropy is defined as $S_l(\rho) = \frac{d}{d-1}(1-\tr[\rho^2])$.

The KCBS inequality \cite{kcbs} is the simplest noncontextuality inequality, in the sense that it requires the minimal number of projective measurements, while the KK inequality \cite{kk9ray} is a first one that can be violated by almost all states except the maximally mixed state. In this paper, we focus on these two inequalities. In Section \ref{MCMS}, the maximal contextuality of mixed states (MCMS) for a fixed linear entropy is presented for each inequality.  Section \ref{spectral} then aims to give a spectral analysis on the overall measurement matrix, endeavoring to explore the question of how the existence of a universal set of measurements depends on the spectral distribution. At last, we give some conclusion and discussions.

\section{The maximally contextual mixed states for the KCBS and KK noncontextuality inequalities}\label{MCMS}

To start with, the KCBS and the KK inequalities are two of the simple and well-known noncontextuality inequalities for three-dimensional systems, usually written as respectively
\[
I_{KCBS} = \sum_{i=1}^5 \langle P_i \rangle \stackrel{\mbox{\tiny{NCHV}}}{\leq} 2 \stackrel{\mbox{\tiny{QM}}}{\leq} \sqrt{5},
\]
and
\[
I_{KK} = \sum_{i=1}^9 \langle P_i \rangle \stackrel{\mbox{\tiny{NCHV}}}{\leq} 3 \stackrel{\mbox{\tiny{QM}}}{\leq} \frac{10}{3},
\]
where $\langle P_i\rangle\equiv\tr[\rho P_i]$, $\rho$ is the general mixed state, and $P_i$'s are rank-$1$ projective measurements with exclusivity relations shown in Figs. \ref{figkcbs} and \ref{figkk}. Note that the former is the simplest inequality that requires the minimal number of measurements, while the latter is a first one that is quantum mechanically violated by all but the maximally mixed state.

\begin{figure}
  \subfigure[]{
    \label{figkcbs}
    \begin{minipage}[b]{0.24\textwidth}
      \centering
 \begin{tikzpicture}[scale=0.45]
 \tikzstyle{every node}=[circle,fill=blue!25,inner sep=0pt,minimum size=0.4cm]
    \foreach \y[count=\a] in {1,3,5,2,4}
      {\pgfmathtruncatemacro{\kn}{72*\y+18}
       \node at (\kn:3) (b\y) {\small \a};}
\draw[] (b1)--(b2)--(b3)--(b4)--(b5)--(b1);
\end{tikzpicture}
    \end{minipage}}%
  \subfigure[]{
    \label{figkk}
    \begin{minipage}[b]{0.24\textwidth}
      \centering
 \begin{tikzpicture}[scale=0.45]
 \tikzstyle{every node}=[circle,fill=blue!25,inner sep=0pt,minimum size=0.4cm]
    \foreach \y[count=\a] in {9,2,1,4,7}
      {\pgfmathtruncatemacro{\kn}{72*\a+18}
       \node at (\kn:3) (b\y) {\small \y};}
    \foreach \y[count=\a] in {8,6,3}
      {\pgfmathtruncatemacro{\kn}{72*\a-54}
       \node at (\kn:4) (b\y) {\small \y};}
\node at (90:2) (b5) {\small 5};
\draw[] (b5)--(b2)--(b1)--(b4)--(b7)--(b5)--(b9)--(b6);
\draw[] (b4)--(b8)--(b6)--(b3)--(b1);
\draw[] (b2)--(b3);
\draw[] (b7)--(b8);
\end{tikzpicture}
    \end{minipage}}
\caption{The exclusivity graph of KCBS (a) and KK (b).}
\end{figure}
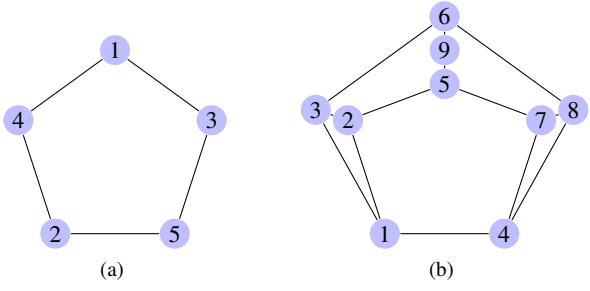

Without loss of generality and for the sake of convenience, we consider a quantum state in the diagonal form:
\[
\rho =
\begin{bmatrix}
\lambda_1 &  & \\
 & \lambda_2 & \\
 & & \lambda_3
\end{bmatrix},
\]
with $\lambda_i$ in decreasing order, and $\vec{\lambda} \equiv (\lambda_1, \lambda_2, \lambda_3)$. This makes sense since for a fixed general state and its optimal measurement set, one can always be able to diagonalize the state and change accordingly its overall measurement set by a global rotation.

In what follows we shall investigate the maximally contextuality $C_q$ of mixed states with respect to a fixed linear entropy $S_l$, for each inequality mentioned above.
 
We plot the upper and lower bounds of contextuality of mixed states in Fig.~\ref{nfig}, together with measuring directions specified in Table \ref{table1} \ref{table2}. In particular, the specific states and the analytic expressions for each curve (except $\bigfrown{AC}$, whose expression is missing) are listed as follows (for convinient, denote the linear entropy $S_l$ as $s$):
\begin{eqnarray}
\bigfrown{AC}:\;\;  \vec{\lambda} &=& (\frac{1+\sqrt{1-{4s}/{3}}}{2}, \frac{1-\sqrt{1-{4s}/{3}}}{2}, 0),\\
\bigfrown{CD}:\;\; \vec{\lambda} &=& (\frac{1+\sqrt{1-s}}{3}, \frac{1+\sqrt{1-s}}{3}, \frac{1-2\sqrt{1-s}}{3}),\nonumber\\
C_q &=& \frac{2\sqrt{1-s}+5}{3},\\
\bigfrown{AD}:\;\; \vec{\lambda} &=& (\frac{1+2\sqrt{1-s}}{3}, \frac{1-\sqrt{1-s}}{3}, \frac{1-\sqrt{1-s}}{3}),\nonumber\\
C_q &=&  \frac{(3 \sqrt{5}-5) \sqrt{1-s}+5}{3},\\
\bigfrown{EF}:\;\;\vec{\lambda} &=& \frac{(3+\sqrt{9-12 s}, 3-\sqrt{9-12 s}, 0)}{6},\nonumber\\
C_q &=& \frac{\sqrt{9-12 s}+57}{18},\\
\bigfrown{FG}:\;\;\vec{\lambda} &=& \frac{(1+\sqrt{4-3 s}, 1, 1-\sqrt{4-3 s})}{3},\nonumber\\
C_q &=& \frac{2 \sqrt{1-s}}{3 \sqrt{3}}+3,\\
\bigfrown{EG}:\;\;\vec{\lambda} &=& \frac{(1+2\sqrt{1-s}, 1-\sqrt{1-s}, 1-\sqrt{1-s})}{3},\nonumber\\
C_q &=& \frac{\sqrt{1-s}+9}{3}.
\end{eqnarray}

\begin{figure}
  \subfigure[]{
    \label{nkcbs} 
    \begin{minipage}[b]{0.25\textwidth}
      \centering
      \includegraphics[width=1.6in]{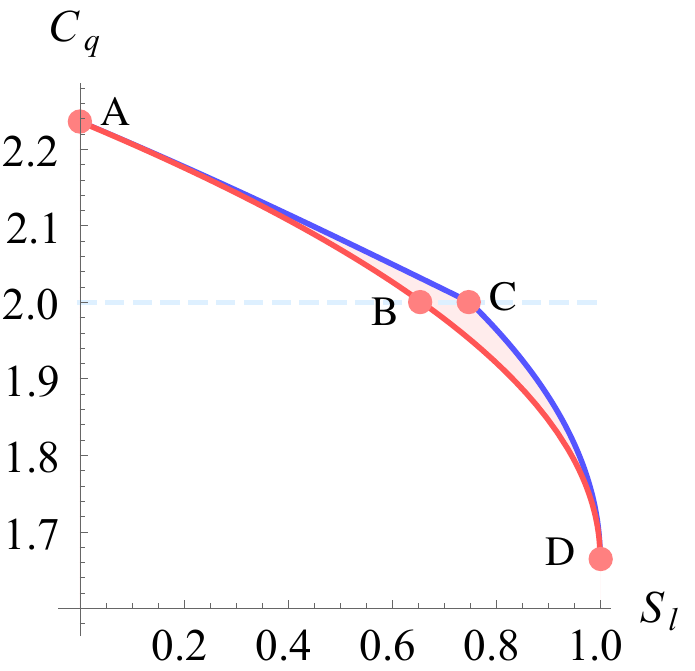}
    \end{minipage}}%
  \subfigure[]{
    \label{nkk} 
    \begin{minipage}[b]{0.25\textwidth}
      \centering
      \includegraphics[width=1.6in]{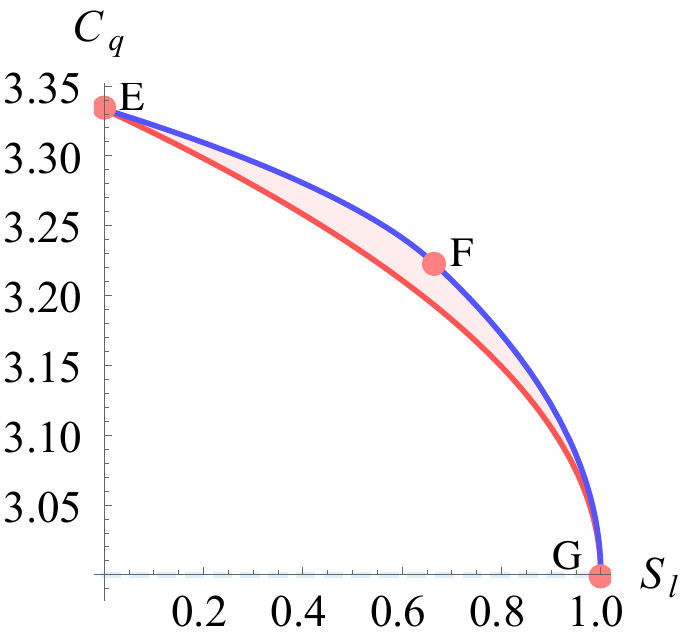}
    \end{minipage}}
  \caption{\scriptsize (Color online) The upper (blue) and lower (red) bounds of contextuality of mixed states for KCBS (a) and KK (b), where $A~(0,\sqrt{5}),~B~(\frac{3}{40} (11-\sqrt{5}),2),~C~(\frac{3}{4},2)$, $D~(1, \frac{5}{3})$, $E~(0,10/3),~F~(2/3,29/9),~G~(1, 3)$.}
  \label{nfig} 
\end{figure}

\begin{table}
\begin{tabular}{llll}
\hline\hline
$|i\rangle$ & $i_1$ & $i_2$ & $i_3$ \\ \hline
$|1\rangle$ & $1$   & $0$   & $0$   \\
$|2\rangle$ & $1$   & $0$   & $0$   \\
$|3\rangle$ & $0$   & $1$   & $0$   \\
$|4\rangle$ & $0$   & $1$   & $0$   \\
$|5\rangle$ & $0$   & $0$   & $1$   \\
\hline\hline
\multicolumn{4}{c}{(a)}\\
\end{tabular}
\hspace{0.25cm}
\begin{tabular}{llll}
\hline\hline
$\tau|i\rangle$ & $i_1$                & $i_2$          & $i_3$          \\ \hline
$\tau|1\rangle$ & $\sqrt{\cos(\beta)}$ & $1$            & $0$            \\
$\tau|2\rangle$ & $\sqrt{\cos(\beta)}$ & $\cos(2\beta)$ & $\sin(2\beta)$ \\
$\tau|3\rangle$ & $\sqrt{\cos(\beta)}$ & $\cos(4\beta)$ & $\sin(4\beta)$ \\
$\tau|4\rangle$ & $\sqrt{\cos(\beta)}$ & $\cos(6\beta)$ & $\sin(6\beta)$ \\
$\tau|5\rangle$ & $\sqrt{\cos(\beta)}$ & $\cos(8\beta)$ & $\sin(8\beta)$ \\
\hline\hline
\multicolumn{4}{c}{(b)}\\
\end{tabular}
\caption{The set of measurements for KCBS, (a) for \arc(C,D) while (b) for \arc(A,C), where $\tau = \frac{1}{\sqrt{1+\cos(\beta)}}$, $\beta = \pi/5$.}
\label{table1}
\end{table}

\begin{table}[h!]
\centering
\begin{tabular}{llll}
\hline\hline
$|i\rangle$ & $i_1$        & $i_2$               & $i_3$                \\ \hline
$|1\rangle$ & $1$          & $0$                 & $0$                  \\
$|2\rangle$ & $0$          & $1$                 & $0$                  \\
$|3\rangle$ & $0$          & $0$                 & $1$                  \\
$|4\rangle$ & $0$          & $1/\sqrt{2}$        & $-1/\sqrt{2}$        \\
$|5\rangle$ & $1/\sqrt{3}$ & $0$                 & $-\sqrt{2}/\sqrt{3}$ \\
$|6\rangle$ & $1/\sqrt{3}$ & $\sqrt{2}/\sqrt{3}$ & $0$                  \\
$|7\rangle$ & $1/\sqrt{2}$ & $1/2$               & $1/2$                \\
$|8\rangle$ & $1/\sqrt{2}$ & $-1/2$              & $-1/2$               \\
$|9\rangle$ & $1/\sqrt{2}$ & $-1/2$              & $1/2$                \\
\hline\hline
\end{tabular}
\caption{The set of measurements for KK.}
\label{table2}
\end{table}

\section{A spectral analysis}\label{spectral}

With the above results, we note that there exist a universal set of measuring directions for the KK inequality; however, this is not true for the KCBS inequality. To proceed, we shall do a spectral analysis on the measuring sets relative to these inequalities.

Define 
\begin{eqnarray}
M= \sum_i P_i,
\end{eqnarray}
which is precisely the expression of inequalities before taking average with a specific state. The quantum prediction can be denoted as $\langle M \rangle = \tr[M\rho]$.

First, we introduce a lemma on Hermitian matrices:
\begin{lemma}\label{lemma}
Assume that 
$$
A = \begin{bmatrix}
a_1 & & & \\
& a_2 & & \\
& & \ddots &\\
& & & a_n\\
\end{bmatrix}, 
B  = U \begin{bmatrix}
b_1 & & & \\
& b_2 & & \\
& & \ddots &\\
& & & b_n\\
\end{bmatrix} U^\dagger, 
$$
where $a_1 \ge a_2 \ge \cdots \ge a_n$, $b_1 \ge b_2 \ge \cdots \ge b_n$, $U = (u_{ij})$ is a unitary matrix. Then 
$$
\tr[AB] \le \vec{a}.\vec{b},
$$
where $\vec{a} = (a_1, a_2, \cdots, a_n),~\vec{b} = (b_1, b_2, \cdots, b_n)$.
\end{lemma}
\begin{proof}
Directly computation shows that
$$
\tr[AB] =  \vec{a} W \vec{b},
$$
where $W = (w_{ij})$ with $w_{ij} = |u_{ij}|^2$. 

Then $W$ is a doubly stochastic matrix by the definition of doubly stochastic matrices.  The Birkhoff-von Neumann theorem says that the set of $n \times n$ doubly stochastic matrices forms a convex polytope whose vertices are the $n \times n$ permutation matrices.  If we consider the linear functional $f(W) = \vec{a} W \vec{b}$ on that  convex polytope, then its optimal can be achieved at the vertices, i.e., the permutation matrices. Since $\vec{a}, \vec{b}$ are already in decreasing order, the maximal of $f(W) = \vec{a}.\vec{b}$ can be achieved when $W = I_n$, which implies $U = I_n$.
\end{proof}

Given a state 
$$
\rho = \begin{bmatrix}
\lambda_1 & &\\
& \lambda_2 &\\
& & \lambda_3\\
\end{bmatrix},
$$
where, $\lambda_1 \ge \lambda_2 \ge \lambda_3$. And the set of measurements $\{|i\rangle\}$ is optimal for $\rho$. Assume $U$ is such a unitary matrix that $U^\dagger \left(\sum_i |i\rangle \langle i| \right) U$ is diagonal and the diagonal is in decreasing order, Lemma \ref{lemma} tells us that the set of measurements $\{ \tilde{|i\rangle} = U^\dagger |i\rangle\}$ is also optimal while $M = \sum_i \tilde{|i\rangle} \tilde{\langle i|}$ is diagonal. Thus, for our purpose, we can only consider the diagonal $M$:
\[
M =
\begin{bmatrix}
m_1 & & \\
 & m_2 & \\
 & & m_3
\end{bmatrix},
\]
with $m_i$ in decreasing order, $m_1 + m_2 + m_3 = n$, and $n$ being the total number of settings (i.e., 5 for KCBS and 9 for KK).

Denote $\vec{m} = (m_1, m_2, m_3)$. The condition $m_1 + m_2 + m_3 = n$ restrains $\vec{m}$ within an, e.g., $m_1m_2$-plane. In general, $C_q = \vec{\lambda}\cdot\vec{m} \le n \lambda_1$. In particular, $C_q = n \lambda_1$, which holds only to an edgeless exclusivity graph $G$, that is, there is no exclusive relation between any pair of events of probability. In fact, the exclusivity relation will further limit the distribution of $\vec{m}$. As we shall see, the exclusivity relations for the KCBS and KK inequalities are so strong that $\vec{m}$ will be dramatically restrained to a curve, rather than a region in the plane.

The curves of $(m_1, m_2)$ for KCBS and KK inequalities are plotted in Fig. \ref{ekcbs} and Fig. \ref{ekk}, respectively. (For comparison, see Fig. \ref{skcbs} for the quantum violation of the KCBS inequality by a convex mixture of $|1\rangle$, $|3\rangle$ and $|5\rangle$ in Table \ref{table1}.) Obviously, $m_2$ must be a function of $m_1$. Then
\begin{eqnarray*}
C_q(m_1)&=& \vec{\lambda}\cdot\vec{m}\nonumber\\
&=& n \lambda_3 + (\lambda_1 - \lambda_3) m_1 + (\lambda_2 - \lambda_3) m_2(m_1).
\end{eqnarray*}
By differentiating the expression with respect to $m_1$,
\[
\frac{d C_q(m_1)}{d m_1} = (\lambda_1 - \lambda_3) + (\lambda_2 - \lambda_3) \frac{dm_2(m_1)}{d m_1},
\]
we will obtain the optimal $M$ for a given state $\rho$.

\begin{figure}
  \subfigure[]{
    \label{ekcbs} 
    \begin{minipage}[b]{0.25\textwidth}
      \centering
      \includegraphics[width=1.6in]{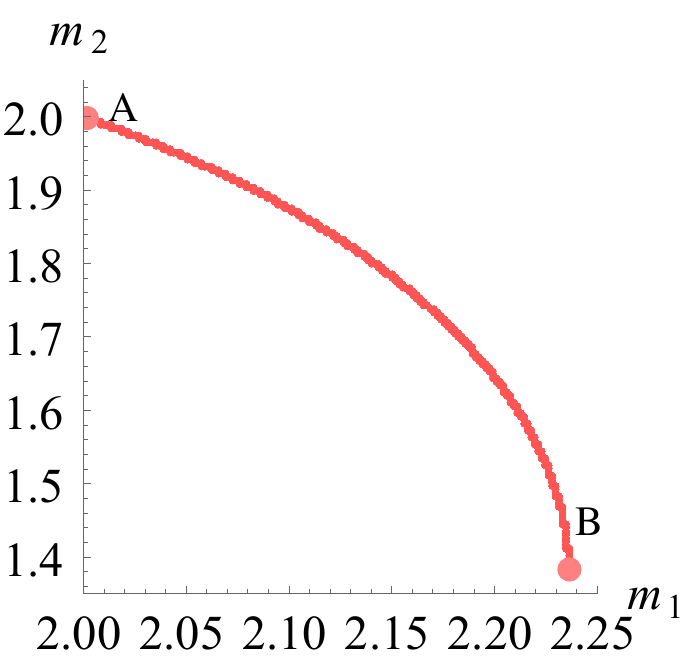}
    \end{minipage}}%
  \subfigure[]{
    \label{ekk} 
    \begin{minipage}[b]{0.25\textwidth}
      \centering
      \includegraphics[width=1.6in]{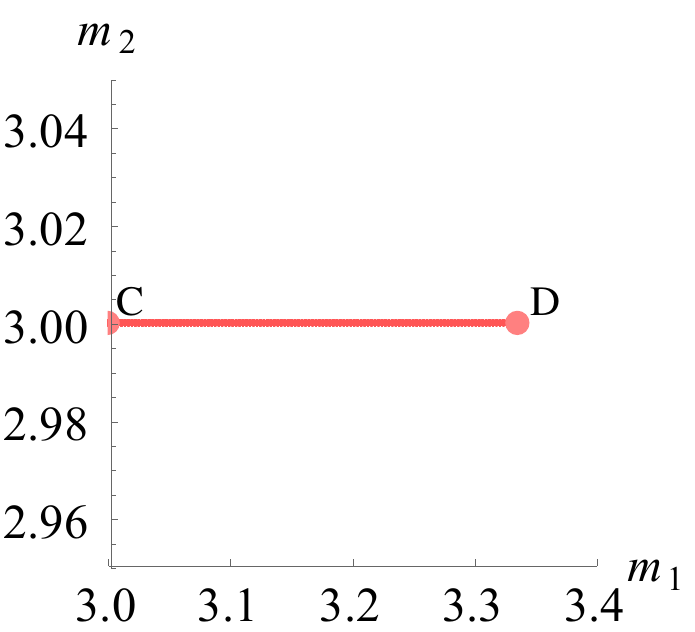}
    \end{minipage}}
  \caption{\scriptsize The curve of $(m_1,m_2)$ for KCBS (a) and KK (b), where $A (2,2)$, $B(\sqrt{5},\frac{5-\sqrt{5}}{2})$, $C (3,3)$, $D (\frac{10}{3},3)$.}
  \label{efig} 
\end{figure}

\begin{figure}[h]
\centering
\includegraphics[width=0.4\textwidth]{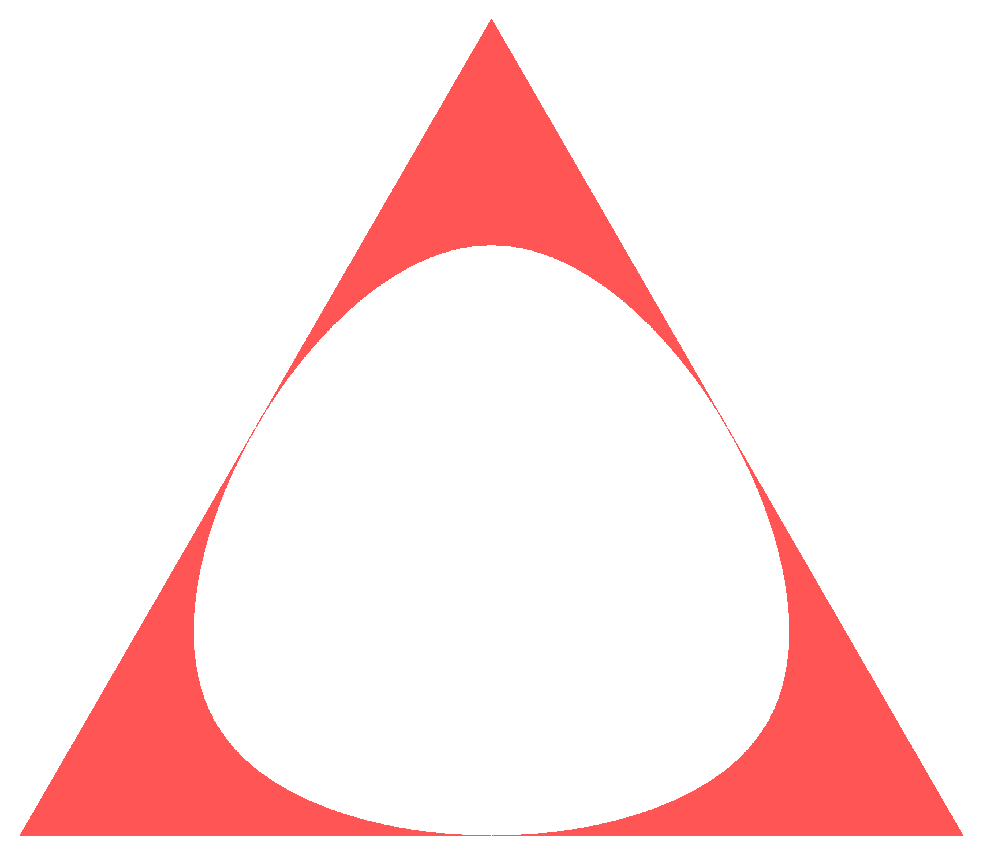}
\caption{\scriptsize The violation of diagonal states for KCBS.}
\label{skcbs}
\end{figure}

For the KCBS inequality, we have $dm_2/dm_1 < -1$, and

\noindent(i):
\[
\bigfrown{AC}:\;\;\frac{d C_q}{d m_1} = \frac{1+\sqrt{1-\frac{4s}{3}}}{2} + \frac{1-\sqrt{1-\frac{4s}{3}}}{2}\frac{dm_2}{d m_1}.
\]
Then $\frac{d C_q}{d m_1} = 0$ implies that
\[
\frac{dm_2}{d m_1} = -\frac{1+\sqrt{1-\frac{4s}{3}}}{1-\sqrt{1-\frac{4s}{3}}}.
\]
Hence, $m_1$ varies with $s$, meaning that there is no universal set of measurements for this curve.

\noindent(ii):
\[
\bigfrown{CD}:\;\;\frac{d C_q}{d m_1} = \sqrt{1-s}(1+\frac{dm_2}{d m_1}) < 0.
\]
This shows that the set of measurements in Table I is optimal.

\noindent(iii):
\[
\bigfrown{AD}:\;\;\frac{d C_q}{d m_1} = \sqrt{1-s} > 0.
\]
Again, this shows that the set of measurements in Table II is optimal.

For the KK inequality, we find that $dm_2/dm_1 = 0$ (i.e., $dc_q/dm_1 \ge 0$) always holds. So all the states that violate the inequality possess the same set of measurements shown in Table III.

Moreover, the condition $m_1 + m_2 + m_3 = n$ yields 
$$\frac{dm_3}{d m_1} = -(1+\frac{dm_2}{d m_1}).$$
For KCBS, this yields $\frac{dm_3}{d m_1} > 0$, a monotonically increasing relation between $m_1$ and $m_3$, implying that they take their maxima simultaneously, while for KK this yields $\frac{dm_3}{d m_1} < 0$, a monotonically decreasing relation, implying that the maximal $m_3$ is obtained when $m_1$ reaches its minimum, and vice versa.

Consequently, the spectral distribution of a noncontextualiy inequality can reflect the nature as to whether there exist a universal set of measurements, so that possible experimental setups could be greatly facilitated.


\section{Conclusion and discussions}\label{conclusion}
In this paper, we have investigated the quantum contextuality of mixed states for the KCBS and the KK noncontextuality inequalities, and explored the question of why there exists a universal set of measurements for the latter, whereas none does for the former inequality. 
We have shown that a spectral analysis on the set of measurements may provide insightful clues toward the ultimate answer to this question.
We believe that further works on combining graph theory and spectral theory in studying quantum contextuality may shed new light on these problems.

\begin{acknowledgments}
J.L.C. is supported by the National Basic Research Program (973 Program) of China under Grant No.\ 2012CB921900 and the NSF of China (Grant Nos.\ 11175089 and 11475089). This
work is also partly supported by the National Research Foundation and the Ministry of Education, Singapore.
\end{acknowledgments}

\end{document}